\documentclass[12pt,a4paper]{article}
\usepackage{a4wide}
\usepackage{microtype}
\usepackage[hidelinks]{hyperref}
\usepackage{amsmath}
\usepackage{amsthm}
\usepackage{mathtools}
\usepackage{etoolbox}
\usepackage[linesnumbered]{algorithm2e}
\usepackage[utf8]{inputenc}
\usepackage{enumerate}
\usepackage{amsmath}
\usepackage{amsfonts}
\usepackage{nicefrac}
\usepackage{amssymb}
\usepackage{amsthm}
\usepackage{mathtools}
\usepackage{tikz}
\usepackage{xcolor}
\usepackage{authblk}
\usepackage{tabularx}
\usepackage{paralist}
\usepackage{booktabs}
\usepackage{paralist}
\usepackage[square,numbers]{natbib}
\usepackage[capitalize]{cleveref}
\usepackage{etoolbox}
\usepackage{xspace}

\newtheorem{theorem}{Theorem}
\newtheorem{lemma}[theorem]{Lemma}

\theoremstyle{definition}

\newcommand{\calT}{\mathcal{T}}

\crefname{problem}{Problem}{Problems}
\Crefname{problem}{Prob.}{Probs.}

\crefname{algocfline}{line}{lines}
\Crefname{algocfline}{l.}{l.}

\usepackage{dsfont}
\newcommand{\N}{\mathbb{N}}

\newcommand{\R}{\mathbb{R}}

\newcommand{\bigO}{\mathcal{O}}

\newcommand{\cocl}[1]{\ensuremath{\operatorname{#1}}}
\newcommand{\NP}{\cocl{NP}}

\newcommand{\calA}{\mathcal{A}}
\newcommand{\calB}{\mathcal{B}}
\newcommand{\calC}{\mathcal{C}}
\newcommand{\calP}{\mathcal{P}}

\newcommand{\calS}{\mathcal{S}}
\newcommand{\eps}{\varepsilon}
\newcommand{\vc}{\mathrm{vc}}

\DeclareMathOperator{\OPT}{OPT}
\DeclareMathOperator{\OPTHS}{OPT_{HS}}
\DeclareMathOperator{\OPTDS}{OPT_{DS}}
\DeclareMathOperator{\OPTCAP}{OPT_{CAP}}
\DeclareMathOperator{\OPTIND}{OPT_{IND}}
\DeclareMathOperator{\OPTCON}{OPT_{CON}}
\DeclareMathOperator{\OPTNST}{OPT_{NST}}

\DeclarePairedDelimiter{\abs}{\lvert}{\rvert}

\newcommand{\CAP}{\mathrm{cap}}

\newcommand{\creflemmapart}[2]{
	\hyperref[#2]{\namecref{#1}~\labelcref*{#1}(\ref*{#2})}
}

\allowdisplaybreaks

\title{Approximate Turing kernelization and lower bounds for domination problems}

\author{Stefan Kratsch}
\author{Pascal Kunz\thanks{Supported by the DFG Research Training Group 2434 ``Facets of Complexity''.}}

\affil{\small
	Algorithm Engineering, Humboldt-Universit\"at zu Berlin, Berlin, Germany.\protect\\
	\texttt{\{kratsch,kunzpasc\}@informatik.hu-berlin.de}}

\date{}

\begin{document}

\maketitle

\begin{abstract}
	An $\alpha$-approximate polynomial Turing kernelization is a polynomial-time algorithm that computes an $(\alpha c)$-approximate solution for a parameterized optimization problem when given access to an oracle that can compute $c$-approximate solutions to instances with size bounded by a polynomial in the parameter.
	Hols et al.~[ESA 2020] showed that a wide array of graph problems admit a $(1+\eps)$-approximate polynomial Turing kernelization when parameterized by the treewidth of the graph and left open whether \textsc{Dominating Set} also admits such a kernelization.
	
	We show that \textsc{Dominating Set} and several related problems parameterized by treewidth do not admit constant-factor approximate polynomial Turing kernelizations, even with respect to the much larger parameter vertex cover number, under certain reasonable complexity assumptions.
	On the positive side, we show that all of them do have a $(1+\eps)$-approximate polynomial Turing kernelization for every $\eps>0$ for the joint parameterization by treewidth and maximum degree, a parameter which generalizes cutwidth, for example.
\end{abstract}

\section{Introduction}

The gold standard in kernelization is a polynomial (exact) kernelization, i.e.~a compression of input instances to a parameterized problem to a size that is polynomial in the parameter such that an exact solution for the original instance can be recovered from the compressed instance.
Several weaker notions of kernelization have been developed for problems that do not admit polynomial kernelizations.
Turing kernelization~\cite{Binkele2012,Hermelin2013} does away with the restriction that the solution must be recovered from a single compressed instance and instead allow several small instances to be created and the solution to be extracted from solutions to all of these instances.
Lossy kernelizations~\cite{Lokshtanov2017}, in turn, do away with the requirement that the solution that can be recovered from the compressed instance be an optimum solution, allowing the solution to the original instance to be worse than optimal by a constant factor.
Hols et al.~\cite{Hols2020} introduced lossy Turing kernelizations, which allow both multiple compressed instances and approximate solutions, and showed that several graph problems parameterized by treewidth admit $(1+\eps)$-approximate Turing kernelizations for every $\eps>0$.
They left as an open question whether or not the problem \textsc{Dominating Set} parameterized by treewidth also admits a constant-factor approximate Turing kernelization.

\subparagraph*{Our contribution.}

We answer this question in the negative and show that a $\bigO(2^{\log ^c \mathrm{vc}})$-approximate polynomial Turing kernelization for $\textsc{Dominating Set}[\mathrm{vc}]$\footnote{We use $\mathrm{FOO}[X]$ to refer to the problem FOO parameterized by $X$.}, where $\mathrm{vc}$ refers to the vertex cover number, would contradict the Exponential Time Hypothesis.
We prove analogous lower bounds for \textsc{Capacitated Dominating Set}[vc], \textsc{Connected Dominating Set}[vc] for \textsc{Hitting Set}[$\abs{U}$], where $U$ is the universe, and for \textsc{Node Steiner Tree}[$\abs{V\setminus T}$] where $V\setminus T$ is the set of non-terminal vertices.
Of course, the lower bounds for the vertex cover number also imply lower bounds for the smaller parameter treewidth.
Using a second approach for obtaining lower bounds for approximate Turing kernelizations, essentially a gap-introducing polynomial parameter transformation (PPT), we show that $\textsc{Independent Dominating Set}[\vc]$ does not have an $\alpha$-approximate polynomial Turing kernelization for any constant $\alpha$, unless every problem in the complexity class $\mathrm{MK}[2]$ has a polynomial (exact) Turing kernelization, which would contradict a conjecture by Hermelin et al.~\cite{Hermelin2013}.
We then show that for the joint parameterization by treewidth and the maximum degree, each of the aforementioned domination problems does have a $(1+\eps)$-approximate polynomial Turing kernelization for every $\eps > 0$.
This generalizes, for instance, parameterization for cutwidth or bandwidth.

\subparagraph*{Related work.}
For an introduction to kernelization, including brief overviews on lossy and Turing kernelizations, we refer to the standard textbook~\cite{Fomin2019}.
Binkele-Raible et al.~\cite{Binkele2012} introduced the first Turing kernelization for a problem that does not admit a polynomial kernelization.
Since then numerous Turing kernelizations have been published for such problems.
Hermelin et al.~\cite{Hermelin2013} introduced a framework, which we will make use of in \cref{sec:mk2-lb}, for ruling out (exact) polynomial Turing kernelizations.
Fellows et al.~\cite{Fellows2018} were the first to combine the fields of kernelization and approximation.
A later study by Lokshtanov et al.~\cite{Lokshtanov2017} introduced the framework of lossy kernelization that has become more established.
Finally, Hols et al.~\cite{Hols2020} gave the first approximate Turing kernelizations.

Approximation algorithms and lower bounds for domination problems have received considerable attention.
They are closely related to the problems \textsc{Hitting Set} and \textsc{Set Cover}.
A classical result by Chv\'atal~\cite{Chvatal1979} implies a polynomial-time $\bigO(\log n)$-factor approximation for \textsc{Dominating Set}.
There are also $\bigO(\log n)$-factor approximations for \textsc{Connected Dominating Set}~\cite{Guha1998} and \textsc{Capacitated Dominating Set}~\cite{Wolsey1982}, but \textsc{Independent Dominating Set} does not have a $\bigO(n^{1-\eps})$-approximation for any $\eps > 0$, unless $\mathrm{P} = \mathrm{NP}$~\cite{Halldorsson1993}.
Chleb\'ik and Chleb\'ikov\'a~\cite{Chlebik2008} showed that 
\textsc{Dominating Set}, \textsc{Connected Dominating Set}, and \textsc{Capacitated Dominating} do not have constant-factor approximations even on graphs with maximum degree bounded by a constant $\Delta$ and that \textsc{Independent Dominating Set} does not have better than a $\Delta$-factor approximation.

\section{Preliminaries}
\label{sec:prelim}
\subparagraph*{Graphs.}
We use standard graph terminology and all graphs are undirected, simple, and finite.
For a graph $G=(V,E)$ and $X\subseteq V$, we use $N[X] \coloneqq X \cup \{v \in V \mid \exists u \in X \colon \{u,v\} \in E\}$ to denote the \emph{closed neighborhood} of $X$ and, if $v \in V$, then we let $N[v] \coloneqq N[\{v\}]$.
We will use $\Delta$ and $\mathrm{vc}$ to refer to the maximum degree and vertex cover number of a graph, respectively.

Let $G=(V,E)$ be a graph $X\subseteq V$ a vertex set.
The set $X \subseteq V$ is a \emph{dominating set} if $N[X] = V$.
It is an \emph{independent set} if there is no edge $\{u,v\} \in E$ with $u,v \in X$.
It is an \emph{independent dominating set} if it is both an independent and a dominating set.
It is a \emph{connected dominating set} if it is a dominating set and the graph $G[X]$ is connected.
A \emph{capacitated graph} $G=(V,E,\CAP)$ consists of a graph $(V,E)$ and a \emph{capacity function} $\CAP \colon V \to \N$.
A \emph{capacitated dominating set} in a capacitated graph $G=(V,E,\CAP)$ is a pair $(X,f)$ where $X \subseteq V$ and $f \colon V\setminus X \to X$ such that (i)~$v$ and $f(v)$ are adjacent for all $v\in V \setminus X$ and (ii)~$\abs{f^{-1}(v)} \leq \CAP(v)$ for all $v\in X$.
The size of $(X,f)$ is $\abs{X}$.

A \emph{tree decomposition} of a graph $G=(V,E)$ is a pair $\calT=(T=(W,F),\{X_t\}_{t\in W})$ where $T$ is a tree, $X_t \subseteq V$ for all $t\in W$, $\bigcup_{t \in W} X_t = V$, for each $e \in E$ there is a $t\in W$ such that $e\subseteq X_t$, and for each $v\in V$ the node set $\{t \in W \mid v \in X_t\}$ induces a connected subgraph of $T$.
The \emph{width} of $\calT$ is $\max_{t \in W} \abs{X_t} -1$.
The \emph{treewidth} $\mathrm{tw}(G)$ of $G$ is the minimum width of any tree decomposition of $G$.
A \emph{rooted tree decomposition} consists of a tree decomposition $\calT=(T=(W,F),\{X_t\}_{t\in W})$ along with a designated root $r \in W$.
Given this rooted decomposition and a node $t \in W$, we will use $V_t\subseteq V$ to denote the set of vertices $v$ such that $v$ such $v \in X_{t'}$ and $t'$ is a descendant (possibly $t$ itself) of $t$ in the rooted tree $(T,r)$.
The rooted tree decomposition is \emph{nice} if $X_r=\emptyset$ and $X_t =\emptyset$ for every leaf of $T$ and every other node $t$ of $T$ is of one of three types: (i)~a \emph{forget node}, in which case $t$ has a single child $t'$ and there is a vertex $v\in V$ such that $X_{t'} = X_t \cup \{v\}$, (ii)~an \emph{introduce node}, in which case $t$ has a single child $t'$ and there is a vertex $v\in V$ such that $X_{t} = X_{t'} \cup \{v\}$, or (iii)~a \emph{join node}, in which case $t$ has exactly two children $t_1$ and $t_2$ and $X_t = X_{t_1} = X_{t_2}$.

\subparagraph*{Turing kernelization.}
A \emph{parameterized decision problem} is a set $L\subseteq \Sigma^* \times \mathbb N$.
A \emph{Turing kernelization} of size $f\colon \mathbb N \to \mathbb N$ for a parameterized decision problem $L$ is a polynomial-time algorithm that receives as input an instance $(x,k) \in \Sigma^* \times \mathbb N$ and access to an oracle that, for any instance $(x',k') \in \Sigma^* \times \mathbb N$ with $\abs{x'} + k' \le f(k)$, outputs whether $(x',k') \in L$ in a single step, and decides whether $(x,k) \in L$.
It is a \emph{polynomial Turing kernel} if $f$ is polynomially bounded.

A \emph{polynomial parameter transformation} (PPT) from one parameterized decision problem~$L$ to a second such problem $L'$ is a polynomial-time computable function $f\colon \Sigma^* \times \mathbb N \to \Sigma^* \times \mathbb N$ such that $(x,k) \in L$ if and only $(x',k') \in L'$ and there is a polynomially bounded function $p$ such that $p(k') \le k$ for all $(x,k),(x',k') \in \Sigma^*\times \mathbb N$ with $f(x,k) = (x',k')$.

A \emph{parameterized minimization problem} is defined by a computable function $\calP \colon \Sigma^* \times \mathbb N \times \Sigma^* \to \mathbb R \cup \{\infty, -\infty\}$.
The optimum value for an instance $(I,k) \in \Sigma^* \times \mathbb N$ is $\OPT_{\calP}(I,k) \coloneqq \min_{x \in \Sigma^*} \calP(I,k,x)$.
We will say that a solution $x \in \Sigma^*$ is \emph{$\alpha$-approximate} if $\calP(I,k,x) \le \alpha \cdot \OPT_{\calP}(I,k)$.
In order to simplify notation, we will allow ourselves to write $\calP(I,x)$ instead of $\calP(I,k,x)$ and $\OPT_{\calP}(I)$ instead of $\OPT_{\calP}(I,k)$ if those values do not depend on $k$.
The problems (\textsc{Capacitated/Connected/Independent}) \textsc{Dominating Set} are defined by $\calP(G,X) \coloneqq \abs{X}$ if $X$ is (capacitated/connected/independent) dominating set in $X$ and $\calP(G,X) \coloneqq \infty$, otherwise.
The problem \textsc{Node Steiner Tree} is defined by $\mathrm{NST}((G,T),X) \coloneqq \abs{X}$ if $G[X\cup T]$ is connected and $\mathrm{NST}((G,T),X) \coloneqq \infty$, otherwise.
\textsc{Hitting Set} is a problem whose input $(U,\calS)$ consists of a set $U$ and a family $\calS \subseteq 2^U$ of nonempty sets, a solution $X$ is a subset of $U$, and $\mathrm{HS}(X) \coloneqq \infty$ if there is an $S \in \calS$ such that $X \cap S = \emptyset$ and $\mathrm{HS}(X) \coloneqq \abs{X}$, otherwise.

Let $\alpha \in \mathbb R$ with $\alpha \ge 1$ and let $\calP$ be a parameterized minimization problem.
An \emph{$\alpha$-approximate Turing kernelization} of size $f \colon \mathbb N \to \mathbb N$ for $\calP$ is a polynomial-time algorithm that given an instance $(I,k)$ computes a $(c\alpha)$-approximate solution when given access to an oracle for $P$ which outputs a $c$-approximate solution to any instance $(I',k')$ with $\abs{I'} + k' \le f(k)$ in a single step.
It is an \emph{$\alpha$-approximate polynomial Turing kernelization} if $f$ is polynomially bounded.
Note that the algorithm is not given access to $c$, the approximation factor of the oracle, and is not allowed to depend on $c$.
In practice, it can also be helpful for the approximate Turing kernelization algorithm to receive a witness for the parameter value~$k$ as input.
Similarly to Hols et al.~\cite{Hols2020}, we will assume that our approximate Turing kernelizations for the parameterization treewidth plus maximum degree are given as input a graph $G$ and a nice tree decomposition of width $\mathrm{tw}(G)$.
Alternatively, one could also use the polynomial-time algorithm due to Feige et al.~\cite{Feige2008} to compute a tree decomposition of width~$\bigO(\sqrt{\log \mathrm{tw}(G)}\cdot \mathrm{tw}(G))$ and then use this tree decomposition.
We will also assume that the given tree decomposition is nice, which is not really a restriction, because there is a polynomial-time algorithm that converts any tree decomposition into a nice tree decomposition without changing the width~\cite{Kloks1994}.

\section{Lower bounds}
\label{sec:lb}

\subsection{Exponential-time hypothesis}
\label{sec:eth-lb}

In the following, we will show that several problems do not have an approximate Turing kernelization assuming the exponential-time hypothesis (ETH).
The proof builds on a proof due to Lokshtanov et al.~\cite[Theorem~12]{Lokshtanov2016} showing that \textsc{Hitting Set} parameterized by the size of the universe does not admit a lossy (Karp) kernelization unless the ETH fails.

Let \textsc{3-CNF-SAT} denote the satisfiability problem for Boolean formulas in conjunctive normal form with at most three literals in each clause.
The \emph{exponential-time hypothesis} (ETH)~\cite{Impagliazzo2001} states that there is there is a fixed $c > 0$ such that \textsc{3-CNF-SAT} is not solvable in time  $2^{cn} \cdot (n + m)^{\bigO(1)}$, where $n$ and $m$ are the numbers of variables and clauses, respectively.

Let $\calP$ and $\calP'$ be parameterized minimization problems and $f\colon \Sigma^* \times \mathbb N \to \R_+$ a real-valued function that takes instances of $\calP$ as input.
An \emph{$f$-approximation-preserving polynomial parameter transformation} ($f$-APPT) from $\calP$ to $\calP'$ consists of two algorithms:
\begin{itemize}
	\item a polynomial-time algorithm $\calA$ (the \emph{reduction algorithm}) that receives as input an instance $(I,k)$ for $\calP$ and outputs an instance $(I',k')$ for $\calP'$ with $k' \le p(k)$ for some polynomially bounded function $p$ and
	\item a polynomial-time algorithm $\calB$ (the \emph{lifting algorithm}) that receives as input the instances $(I,k)$, $(I',k')$, where the latter is the output of $\calA$ when given the former, as well as a solution $x$ for $(I,k)$ and outputs a solution $y$ for $(I',k')$ with \[\calP(I,k,y) \le \frac{f( I,k ) \cdot \OPT_\calP(I,k) \cdot \calP'(I',k,x)}{\OPT_{\calP'}(I',k)}.\]
\end{itemize}

We will use the following lemma, which is a weaker version of a result by Nelson~\cite{Nelson2007}.

\begin{lemma}[\cite{Nelson2007}]
	\label{lemma:nelson}
	If there are a constant $c<1$ and a polynomial-time algorithm that computes an $\bigO(2^{\log^c \abs{U}})$-factor approximation for \textsc{Hitting Set}, then the ETH fails.
\end{lemma}

\begin{lemma}
	\label{lemma:transfer}
	Let $\calP$ be a parameterized minimization problem that satisfies the following two conditions:
	\begin{enumerate}[(a)]
		\item There is an $f$-APPT from $\textsc{Hitting Set}[\abs{U}]$ to $\calP$ with $f(U,\calS) \in \bigO(2^{\log^{c_1} \abs{U}})$ for some constant $c_1 <1$.
		\item There is a constant $c_2<1$ and a polynomial-time algorithm that computes a $\bigO(2^{\log ^{c_2} \abs{I}})$-factor approximation for $\calP$.
	\end{enumerate}
	Then, there is no $\bigO(2^{\log ^{c_3} k})$-approximate polynomial Turing kernelization for $\calP$ for any $c_3<1$, unless the ETH fails. 
\end{lemma}
\begin{proof}
	Suppose that $\calP$ satisfies conditions (a) and (b) and admits a $\bigO(2^{\log ^{c_3} k})$-approximate polynomial Turing kernelization of size $\bigO(k^d)$.
	Furthermore, assume that $p(n) \le \bigO(n^{d'})$ where $p$ is the polynomial parameter bound for the reduction algorithm of the $f$-APPT.
	Choose any constant $c$ with $\max\{c_1,c_2,c_3\} < c < 1$ and observe that for any constant $\alpha$ and $i \in\{1,2,3\}$ we have that $2^{\alpha \log^{c_i} n } \le \bigO(2^{\log^c n})$.
	We will give a $\bigO(2^{\log^c \abs{U}})$-approximation algorithm for \textsc{Hitting Set}.
	By \cref{lemma:nelson}, this proves the claim.
	
	The algorithm proceeds as follows.
	Given an instance $I=(U,\calS)$ of \textsc{Hitting Set} as input, it first applies the reduction algorithm of the APPT to obtain an instance $(I',k)$ of $P$.
	Then, it runs the given approximate Turing kernelization on $(I',k)$.
	Whenever this Turing kernelization queries the oracle, this query is answered by running the approximation algorithm given by condition~(b).
	Once the Turing kernelization outputs a solution $X$, the algorithm calls the lifting algorithm of the APPT on $X$, $(I,\abs{U})$, and $(I',k)$.
	The algorithm outputs the solution $Y$ given by the lifting algorithm.
	
	It remains to show that $\abs{Y} \le \bigO (2^{\log^{c_3} \abs{U}} \cdot \OPTHS(I))$.
	First, observe that the $\bigO(2^{\log ^{c_2} \abs{I}})$-factor approximation algorithm is only run on instances $(J,\ell)$ with $\abs{J} \le \bigO(k^d)$, so it always outputs a solution $Z$ with $\calP(J,\ell,Z) \le \bigO(2^{\log ^{c_2} k^d}\cdot\OPT_\calP(J,\ell))$.
	Hence, in the algorithm described above the Turing kernelization is given a $\bigO(2^{d\log ^{c_2} k})$-approximate oracle, so it follows that $\calP(I',k,X) \le \bigO(2^{\log ^{c_2} k} \cdot 2^{d \log ^{c_3} k} \cdot \OPT_\calP(I',k))$.
	Since $k \le \bigO(\abs{U}^{d'})$, it follows that $\calP(I',k,X) \le \bigO(2^{\log ^{c_3} \abs{U}^{d'}} \cdot 2^{d \log ^{c_2} \abs{U}^{d'}} \cdot \OPT_\calP(I',k))$.
	Therefore:
	\begin{align*}
	\abs{Y} & \le \frac{f( I,k ) \cdot \OPTHS(I) \cdot \calP(I',k,X)}{\OPT_{\calP}(I',k)} \le \bigO \left( 2^{\log ^{c_3} \abs{U}^{d'}} \cdot 2^{d \log ^{c_2} \abs{U}^{d'}} \cdot f(I,k) \cdot \OPTHS(I)  \right)\\
	& \le \bigO \left( 2^{\log ^{c_3} \abs{U}^{d'}} \cdot 2^{d \log ^{c_2} \abs{U}^{d'}} \cdot f(I) \cdot \OPTHS(I)  \right)\\
	& \le \bigO \left( 2^{\log ^{c_3} \abs{U}^{d'}} \cdot 2^{d \log ^{c_2} \abs{U}^{d'}} \cdot 2^{\log^{c_1} \abs{U}} \cdot \OPTHS(I)  \right) \\
	& \le \bigO(2^{\log^{c} \abs{U}} \cdot \OPTHS(I)). \qedhere 
	\end{align*}
\end{proof}

With \cref{lemma:transfer}, we can prove approximate Turing kernelization lower bounds for several parameterized minimization problems.

\begin{theorem}
	\label{thm:lb-eth}
	Unless the ETH fails, there are no $\bigO(2^{\log ^c k})$-approximate polynomial Turing kernels, for any $c<1$ and where $k$ denotes the respective parameter, for the following parameterized minimization problems:
	\begin{enumerate}[(i)]
		\item $\textsc{Hitting Set}[\abs{U}]$,
		\item $\textsc{Dominating Set}[\vc]$,
		\item $\textsc{Capacitated Dominating Set}[\vc]$,
		\item $\textsc{Connected Dominating Set}[\vc]$, and
		\item $\textsc{Node Steiner Tree}[V\setminus T]$.
	\end{enumerate}
\end{theorem}
\begin{proof}
	For each problem, we will prove conditions (a) and (b) from \cref{lemma:transfer}.
	\begin{enumerate}[(i)]
		\item
		\begin{enumerate}[(a)]
			\item Immediate.
			\item Chv{\'{a}}tal~\cite{Chvatal1979} gives a $\bigO(\log \abs{\calS})$-factor approximation algorithm.
		\end{enumerate}
		\item
		\begin{enumerate}[(a)]
			\item The following folklore reduction is a $1$-APPT.
			Let $(U,\calS)$ be an instance of hitting set.
			The algorithm $\calA$ creates a graph $G$ as follows.
			For every $x \in U$ and for every $S \in \calS$, $G$ contains vertices $v_x$ and $w_S$, respectively, and $G$ also contains an additional vertex $u$.
			The vertices $\{v_x \mid x \in U\} \cup \{u\}$ form a clique and there is an edge between $v_x$ and $w_S$ if and only if $x\in S$.
			Observe that the vertices $\{v_x \mid x \in U\}$ form a vertex cover in $G$, so clearly $\mathrm{vc}(G) \le \abs{U}$, and that $\OPTHS(U,\mathcal S) = \OPTDS(G)$.
			
			Let $X$ be a dominating set in $G$.
			Let $X'$ be obtained from $X$ by removing $z$ and replacing any $w_S$ by an arbitrary $v_x$ with $x\in S$ (such an element $u$ must exist, as we assume that all $S \in \mathcal{S}$ are non-empty).
			The algorithm $\calB$ outputs $Y \coloneqq \{x \in U \mid v_x \in X'\}$.
			This set is a hitting set, because for any $S \in \mathcal{S}$ one of the following cases applies: (i) $w_S \notin X$, meaning that $X$ contains a neighbor $v_x$ of $w_S$. Then, also $x\in X'$ and, hence $x \in Y$ and $x\in S$. (ii) $w_S \in X$, meaning that $w_S$ is replaced by $v_x$ with $x\in S$ when creating $X'$. Then $x\in Y$ and $x\in S$.
			Finally,
			$\abs{Y} = \abs{X} = \frac{\OPTHS(U,\calS) \cdot \abs{X}}{\OPTDS(G)}$,
			since $\OPTHS(U,\calS) = \OPTDS(G)$.
			\item The $\bigO(\log n)$-factor approximation for \textsc{Hitting Set}~\cite{Chvatal1979} can also be used in a straightforward manner to approximate \textsc{Dominating Set}.
		\end{enumerate}
		\item
		\begin{enumerate}[(a)]
			\item Any instance of \textsc{Dominating Set} can be transformed into an equivalent instance of \textsc{Capacitated Dominating Set} by setting $\CAP(v) \coloneqq \deg (v)$ for all vertices $v$.
			The claim then follows by the same argument as for \textsc{Dominating Set}.
			\item Wolsey~\cite{Wolsey1982} gives a $\bigO(\log \abs{\calS})$~factor approximation for \textsc{Capacitated Hitting Set} which can be adapted to approximate \textsc{Capacitated Dominating Set}. 
		\end{enumerate}
		\item
		\begin{enumerate}[(a)]
			\item The APPT given in (ii) for \textsc{Dominating Set} also works for \textsc{Connected Dominating Set}, because, in the graphs produced by $\calA$, $\OPTCON(G) = \OPTDS(G)$ and the solution output by $\calB$ is always a clique and, therefore, connected.
			\item Guha and Khuller~\cite{Guha1998} give a $\bigO(\log \Delta) \le \bigO(\log n)$-factor approximation for \textsc{Connected Dominating Set}.
		\end{enumerate}
		\item
		\begin{enumerate}[(a)]
			\item The following reduction is essentially the same as the one given by Dom et al.~\cite{Dom2014}
			Let $(U,\calS)$ be an instance of \textsc{Hitting Set}.
			The algorithm $\calA$ creates the graph $G$ as in the reduction for \textsc{Dominating Set} in (ii) and sets $T \coloneqq \{w_s \mid s \in \calS\} \cup \{u\}$.
			Clearly, $\abs{V \setminus T} = \abs{U}$ and it easy to show that $\OPTHS (U,\calS) = \OPTNST(G,T)$. By a similar argument as in (ii), the algorithm $\calB$ can output $\{x \in U \mid v_x \in X\}$ where $X$ is a given solution for the \textsc{Node Steiner Tree} instance $(G,T)$.
			
			\item Klein and Ravi~\cite{Klein1995} give a $\bigO(\log n)$-factor approximation for this problem.
			\qedhere
		\end{enumerate}
	\end{enumerate}
\end{proof}

If $\calC$ is a hereditary class of graphs, then we may define the \textsc{Restricted $\calC$-Deletion} problem as follows:
We are given a graph $G=(V,E)$ and $X\subseteq V$ such that $G- X \in \calC$ and are asked to find a minimum $Y \subseteq X$ such that $G- Y \in \calC$.
For \textsc{Restricted Perfect Deletion}[$\abs{X}$], \textsc{Restricted Weakly Chordal Deletion}[$\abs{X}$], and \textsc{Restricted Wheel-Free Deletion}[$\abs{X}$], reductions given by Heggernes et al.~\cite{Heggernes2013} and Lokshtanov~\cite{Lokshtanov2008} can be shown to be $1$-APPTs.
However, it is open whether they have a $\bigO(2^{\log ^c \abs{I}})$-factor approximation with $c<1$, so we cannot rule out an approximate polynomial Turing kernelization.
However, we can observe that, under ETH, they cannot have both an approximation algorithm with the aforementioned guarantee \emph{and} an approximate polynomial Turing kernelization.

More generally, we can deduce from the proof of \cref{lemma:transfer} the following about any parameterized minimization problem $\calP$ that only satisfies the first condition in \cref{lemma:transfer}:
If we define an \emph{approximate polynomial Turing compression} of a problem $\calP$ into a problem $\calP'$ to be essentially an approximate polynomial Turing kernelization for $\calP$, except that it is given access to an approximate oracle for $\calP'$ rather than $\calP$, then we can rule out (under ETH) an approximate polynomial Turing compression of any problem $\calP$ that satisfies the first condition into any problem $\calP'$ that satisfies the second condition in the same lemma.

\subsection{\texorpdfstring{$\boldsymbol{\mathrm{MK}[2]}$-hardness}{MK[2]-hardness}}
\label{sec:mk2-lb}
The approach described in \cref{sec:eth-lb} is unlikely to work for the problem \textsc{Independent Dominating Set}.
That approach requires a $\bigO(2^{\log^c n})$-factor approximation algorithm with $c<1$ to answer the queries of the Turing kernelization.
However, there is no $\bigO(n^{1-\eps})$-factor approximation for this problem for any $\eps>0$ unless $\mathrm{P}=\mathrm{NP}$~\cite{Halldorsson1993}.

In the following, we will prove that there is no constant-factor approximate polynomial Turing kernelization for \textsc{Independent Dominating Set}[vc], assuming a conjecture by Hermelin et al.~\cite{Hermelin2013} stating that parameterized decision problems that are hard for the complexity class $\mathrm{MK}[2]$ do not admit polynomial (exact) Turing kernelizations.

Let \textsc{CNF-SAT} denote the satisfiability problem for Boolean formulas in conjunctive normal form.
The class $\mathrm{MK}[2]$ may be defined as the set of all parameterized problems that can be reduced with a PPT to $\textsc{CNF-SAT}[n]$  where $n$ denotes the number of variables.\footnote{This is not directly the definition given by Hermelin et al.~\cite{Hermelin2013}, but an equivalent characterization.}

We will prove that an $\alpha$-approximate polynomial Turing kernelization for \textsc{Independent Dominating Set}[vc] implies the existence of a polynomial Turing kernelization for \textsc{CNF-SAT}[$n$].
For this, we will need the following lemma allowing us to translate queries between oracles for \textsc{Independent Dominating Set} and \textsc{CNF-SAT} using a standard self-reduction:

\begin{lemma}
	\label{lemma:self-reduction}
	There is a polynomial-time algorithm that, given as input a graph $G$ and access to an oracle that decides in a single step instances of \textsc{CNF-SAT} whose size is polynomially bounded in the size of $G$, outputs a minimum independent dominating set of $G$.
\end{lemma}
\begin{proof}
	The decision version of \textsc{Independent Dominating Set}, in which one is given a graph $H$ and an integer $k$ and is asked to decide whether $H$ contains an independent dominating set of size at most $k$, is in \NP{} and \textsc{CNF-SAT} is \NP-hard, so there is a polynomial-time many-one reduction from \textsc{Independent Dominating Set} to \textsc{CNF-SAT}.
	For any graph $H$ and integer $k$ let $R(H,k)$ denote the instance of \textsc{CNF-SAT} obtained by applying this reduction to $(H,k)$.
	
	Let $n$ be the number of vertices in $G$.
	We first determine the size of a minimum independent dominating set in $G$ by querying the oracle for \textsc{CNF-SAT} on the instance $R(G,k)$ for each $k \in [n]$.
	Observe that the size of $R(G,k)$ is polynomially bounded in the size of $G$.
	Hence, we may input this instance to the oracle.
	Let $k_0$ be the smallest value of $k$ for which this query returns yes.
	
	We must then construct an independent dominating set of size $k_0$ in $G$.
	We initially set $\ell \coloneqq k_0$, $H\coloneqq G$, and $S \coloneqq \emptyset$ and perform the following operation as long as $\ell > 0$.	
	For each vertex $u$ in $H$, we query the oracle on the instance $R(H-N[u],\ell-1)$.
	If this query returns yes, then we add $u$ to $S$ and set $H \coloneqq H - N[u]$ and $\ell \coloneqq \ell -1$.
	Once $\ell = 0$, we output $S$.
	
	We claim that this procedure returns an independent dominating set of size $k_0$ if $k_0$ is the size of a minimum independent dominating set in $G$.
	Let $X$ be a minimum independent dominating set in $G$ and $u \in X$.
	Then, $X \setminus \{u\}$ is a minimum dominating set of size $k_0-1$ in $G - N[u]$ and the claim follows inductively.
\end{proof}

\begin{theorem}
	\label{thm:lb-ind}
	If, for any $\alpha \ge 1$, there is an $\alpha$-approximate polynomial Turing kernelization for \textsc{Independent Dominating Set}[vc], then there is a polynomial Turing kernelization for \textsc{CNF-SAT}[$n$].
\end{theorem}
\begin{proof}
	Assume that there is an $\alpha$-approximate Turing kernelization for \textsc{Independent Dominating Set}[vc] whose size is bounded by the polynomial~$p$.
	We will give a polynomial Turing kernelization for \textsc{CNF-SAT}[$n$].
	
	Let the input be a formula $F$ in conjunctive normal form over the variables $x_1,\ldots,x_n$ consisting of the clauses $C_1,\ldots,C_m$.
	First, we compute a graph $G$ on which we then run the approximate Turing kernelization for \textsc{Independent Dominating Set}.
	The construction of the graph $G$ in the following is due to Irving~\cite{Irving1991}.
	
	Let $s \coloneqq \lceil \alpha \cdot n \rceil+1$.
	The graph $G=(V,E)$ contains vertices $v_1,\ldots,v_n$ and $\overline{v_1},\ldots,\overline{v_n}$, representing the literals that may occur in $F$.
	Additionally, for each $j \in [m]$, there are $s$ vertices $w_j^1,\ldots,w_j^s$ representing the clause $C_j$.
	For each $i \in [n]$, there is an edge between $v_i$ and $\overline{v_i}$.
	There is also an edge between $v_i$ and $w_j^\ell$ for all $\ell \in [s]$ if $x_i \in C_j$ and an edge between $\overline{v_i}$ and $w_j^\ell$ for all $\ell \in [s]$ if $\neg x_i \in C_j$.
	The intuition behind this construction is as follows:
	In $G$ any independent dominating set may contain at most one of $v_i$ and $\overline{v_i}$ for each $i \in [n]$, so any such set represents a partial truth assignment of the variables $x_1,\ldots,x_n$.
	If $F$ is satisfiable, then $G$ contains an independent dominating set of size $n$.
	Conversely, if $F$ is not satisfiable, then any independent dominating set must contain $w^1_j,\ldots,w^s_j$ for some $j\in [m]$, so it must have size at least $s > \alpha \cdot n$, thus creating a gap of size greater than $\alpha$ between yes and no instances.
	Moreover, $\{v_i,\overline{v_i} \mid i \in [n]\}$ is a vertex cover in $G$, so the vertex cover number of $G$ is polynomially bounded in $n$.
	
	The Turing kernelization for \textsc{CNF-SAT} proceeds in the following manner.
	Given the formula $F$, it first computes the graph $G$ and runs the $\alpha$-approximate Turing kernelization for \textsc{Independent Dominating Set} on $G$.
	Whenever the $\alpha$-approximate Turing kernelization queries the oracle on an instance~$G'$, this query is answered using the algorithm given by \cref{lemma:self-reduction}.
	Observe that the size of $G'$ is polynomially bounded in the vertex cover number of $G$, which in turn is polynomially bounded in $n$, so the oracle queries made by this algorithm are possible.
	Let $X$ be the independent dominating set for $G$ output by the approximate Turing kernelization.
	Since this Turing kernelization is given access to a $1$-approximate oracle, $\abs{X} \le \alpha \cdot \OPTIND(G)$.
	We claim that $F$ is satisfiable if and only if $\abs{X} \le \alpha\cdot n$.
	With this claim, the Turing kernelization can return yes if and only if this condition is met.
	
	If $F$ is satisfiable and $\varphi \colon \{x_1,\ldots,x_n\} \to \{0,1\}$ is a satisfying truth assignment, then $Y \coloneqq \{v_i \mid \varphi(x_i) = 1 \} \cup \{\overline{v_i} \mid \varphi(x_i) = 0\}$ is an independent dominating set in $G$.
	Hence, $\abs{X} \le \alpha \cdot \OPTIND(G) \le \alpha \cdot \abs{Y} = \alpha \cdot n$.
	Conversely, suppose that $F$ is not satisfiable.
	For each $i \in [n]$, the set $X$ may contain at most one of the vertices $v_i$ and $\overline{v_i}$.
	Consider the partial truth assignment with $\varphi (x_i) \coloneqq 1$ if $v_i \in X$ and $\varphi(x_i) \coloneqq 0$ if $\overline{v_i}\in X$.
	Because $F$ is unsatisfiable there is a clause $C_j$ that is not satisfied by $\varphi$.
	Hence, $X$ must contain the vertices $w_j^1,\ldots,w_j^s$.
	Therefore, $\abs{X} \ge s > \alpha \cdot n$.
\end{proof}

\section{\texorpdfstring{Turing kernelizations for parameter $\boldsymbol{\mathrm{tw} + \Delta}$}{Turing kernelizations for parameter treewidth plus maximum degree}}

In this section, we will prove that the domination problems for which we proved lower bounds when parameterized by the vertex cover number do have $(1+\eps)$-approximate polynomial Turing kernels when parameterized by treewidth plus maximum degree.
The following lemma is a generalization of \cite[Lemma~11]{Hols2020} and can be proved in basically the same way.

\begin{lemma}
	\label{lemma:find-node}
	Let $G$ be a graph with $n$ vertices, $\calT$ be a nice tree decomposition of width $w$, and $s \le n$.
	Then, there is a node $t$ of $\calT$ such that $s \le \abs{V_t} \le 2s$.
	Moreover, such a node $t$ can be found in polynomial time. 
\end{lemma}

\subsection{Dominating Set}
We start with \textsc{Dominating Set}.

\begin{lemma}
	\label{lemma:ds}
	Let $G=(V,E)$ be a graph.
	\begin{enumerate}[(i)]
		\item\label{lemma:dsI} If $\Delta$ is the maximum degree of $G$, then $\OPTDS(G) \ge \frac{\abs{V}}{\Delta +1}$.
		\item\label{lemma:dsII} If $A,B,C \subseteq V$ with $A\cup C = V$, $A\cap C = B$, and there are no edges between $A\setminus B$ and $C \setminus B$, then		
		$\OPTDS(G) \ge \OPTDS(G[A]) + \OPTDS(G[C]) - 2\abs{B}$.
	\end{enumerate}
\end{lemma}
\begin{proof}
	\begin{enumerate}[(i)]
		\item Every vertex can only dominate its at most $\Delta$ neighbors and itself.
		\item Let $X$ be a dominating set in $G$ of size $\OPTDS(G)$.
		Then $Y \coloneqq (X\cap A) \cup B$ and $Z \coloneqq (X\cap C) \cup B$ are dominating sets in $G[A]$ and $G[C]$, respectively.
		Hence,
		\begin{align*}
		\OPTDS(G) &= \abs{X} = \abs{X\cap A} + \abs{X\cap C} - \abs{X\cap B} \\
		&\ge \abs{Y} - \abs{B\setminus X} + \abs{Z} - \abs{B}\\
		&\ge \OPTDS(G[A]) + \OPTDS(G[C]) - 2\abs{B}.\qedhere
		\end{align*}
	\end{enumerate}
\end{proof}

The kernelization algorithm for \textsc{Dominating Set} uses \cref{lemma:ds}, but is simpler and otherwise similar to the one for \textsc{Capacitated Dominating Set} and is deferred to the appendix.

\begin{theorem}
	\label{thm:ub-ds}
	For every $\eps > 0$, there is a $(1+\eps)$-approximate Turing kernelization for \textsc{Dominating Set} with $\bigO(\frac{1+\eps}{\eps} \cdot \mathrm{tw} \cdot \Delta)$ vertices.
\end{theorem}
\begin{algorithm}[t]
	\DontPrintSemicolon
	\SetKwInOut{Input}{input}\SetKwInOut{Output}{output}
	\Input{A graph $G=(V,E)$, nice tree decomposition $\calT$ of width $\mathrm{tw}$, $\eps>0$}
	$s \leftarrow 2\cdot\frac{1+\eps}{\eps} \cdot (\mathrm{tw}+1) \cdot (\Delta+1)$\;
	\eIf{$\abs{V} \leq s$}{
		Apply the $c$-approximate oracle to $G$ and output the result.\; \label{l:ds1}
	}{
		 Use \cref{lemma:find-node} to find a node $t$ in $\calT$ such that $s \le \abs{V_t} \le 2s$.\;
		 Apply the $c$-approximate oracle to $G[V_t]$ and let $S_t$ be the solution output by the oracle.\;
		 Let $\calT'$ be the tree obtained by deleting the subtree rooted at $t$ except for the node $t$ from $\calT$.\;
		 Apply this algorithm to $(G - (V_t\setminus X_t),\calT',\eps)$ and let $S'$ be the returned solution.\;
		 Return $S'_t \cup S'$.\;\label{l:ds2}
	}

	\caption{A $(1+\eps)$-approximate polynomial Turing kernelization for \textsc{Dominating Set} parameterized by $\mathrm{tw}+\Delta$}
	\label{alg:ds}
\end{algorithm}
\begin{proof}	
	Consider \cref{alg:ds}.
	This algorithm always returns a dominating set of $G$.
	If the algorithm terminates in line~\ref{l:ds1}, then this is true because the oracle always outputs a dominating set.
	If it terminates in line~\ref{l:ds2}, then let $v \in V$ be an arbitrary vertex.
	If $v \in V_t$, then $v$ is dominated by a vertex in $S_t$, because $S_t$ is a dominating set in $G[V_t]$.
	If $v \in V\setminus V_t$, then $v$ is dominated by a vertex in $S'$.
	
	The algorithm runs in polynomial time.
	
	Finally, we must show that the solution output by the algorithm contains at most $c\cdot (1+\eps)\cdot\OPTDS(G)$ vertices.
	We prove the claim by induction on the number of recursive calls.
	If there is no recursive call, the algorithm terminates in line~\ref{l:ds1} and the solution contains at most $c\cdot\OPTDS(G)$ vertices.
	Otherwise, by induction:
	\begin{align*}
	\abs{S'_t \cup S'} &\le \abs{S'_t} + \abs{S'} 
	\le c\cdot\OPTDS(G[V_t]) + \abs{S'}\\
	&= c\cdot(1+\eps)\cdot\OPTDS(G[V_t]) - c\cdot \eps \cdot \OPTDS(G[V_t]) + \abs{S'}\\
	&\stackrel{\dag}{\le} c\cdot(1+\eps)\cdot\OPTDS(G[V_t]) - \frac{c\cdot \eps \cdot \abs{V_t}}{\Delta + 1} + \abs{S'}\\ 
	&\le c\cdot(1+\eps)\cdot\OPTDS(G[V_t]) -  2\cdot c\cdot(1+\eps)\cdot(\mathrm{tw}+1) + \abs{S'}\\
	&\le c\cdot(1+\eps)\cdot\OPTDS(G[V_t]) - 2\cdot c\cdot(1+\eps)\cdot\abs{X_t} + \abs{S'}\\
	& \le c\cdot(1+\eps)\cdot\OPTDS(G[V_t]) - 2\cdot c\cdot(1+\eps)\cdot\abs{X_t} + c\cdot(1+\eps)\cdot\OPTDS(G-V_t)\\
	& = c\cdot(1+\eps)\cdot(\OPTDS(G[V_t]) - 2\abs{X_t} + \OPTDS(G-V_t))\\
	&\stackrel{\ddag}{\le} c\cdot(1 + \eps) \cdot \OPTDS(G).
	\end{align*}
	Here, the inequality marked $\dag$ follows from \creflemmapart{lemma:ds}{lemma:dsI} and the one marked $\ddag$ follows from \creflemmapart{lemma:ds}{lemma:dsII} with $A=(V\setminus V_t)\cup X_t$, $B=X_t$, and $C=V_t \setminus X_t$.
\end{proof}

\subsection{Capacitated Dominating Set}
Next, we consider \textsc{Capacitated Dominating Set}.

\begin{lemma}
	\label{lemma:capds}
	Let $G = (V,E,\CAP)$ be a capacitated graph with maximum degree $\Delta$ and $A,B,C \subseteq V$ such that $A\cup C = V$, $A\cap C = B$, and there are no edges from $A\setminus B$ to $C \setminus B$.
	\begin{enumerate}[(i)]
		\item\label{lemma:capdsI} $\OPTCAP(G) \ge \OPTCAP(G[A]) + \OPTCAP(G[C]) - 2\abs{B}$.
		\item\label{lemma:capdsII} Given capacitated dominating sets $(X,f)$ and $(Y,g)$ in $G[A]$ and $G[C]$, respectively, one can in construct in polynomial time a capacitated dominating set for $G$ of size at most $\abs{X} + \abs{Y} + (\Delta+1) \cdot \abs{B}$.
	\end{enumerate}
\end{lemma}
\begin{proof}
	\begin{enumerate}[(i)]
		\item Let $(X,f)$ with $X\subseteq V$ and $f\colon V\setminus X \to X$ be a capacitated dominating set of size $\OPTCAP(G)$.
		Then, $(Y,g)$ with $Y\coloneqq (X \cap A)\cup B$ and $g(v) \coloneqq f(v)$ for all $v\in A \setminus Y$ is a capacitated dominating set in $G[A]$ and $(Z,h)$ with $Z\coloneqq (X \cap C)\cup B$ and $h(v) \coloneqq f(v)$ for all $v\in C \setminus Z$ is a capacitated dominating set in $G[C]$.
		Hence,
		\begin{align*}
		\OPTCAP(G) &= \abs{X} = \abs{X\cap A} + \abs{X\cap C} - \abs{X\cap B}\\
		&= \abs{Y} - \abs{B\setminus X} + \abs{Z} - \abs{B} - \abs{X \cap B}\\
		&\ge \abs{Y} + \abs{Z} - 2\abs{B}\\
		&\ge \OPTCAP(G[A]) + \OPTCAP(G[C]) - 2\abs{B}.
		\end{align*}
		\item We construct the capacitated dominating set $(Z,h)$ for $G$ as follows.
		Let $Z\coloneqq X \cup Y \cup N[B]$.
		Observe that $\abs{N[B]} \le (\Delta + 1)\abs{B}$.
		Define $h$ by setting $h(v) \coloneqq f(v)$ for all $v \in A \setminus Z$ and $h(v) \coloneqq g(v)$ for all $v \in C \setminus Z$.
		One can easily verify that this is a capacitated dominating set.
		\qedhere
	\end{enumerate}
\end{proof}

\begin{theorem}
	For every $\eps > 0$, there is a $(1+\eps)$-approximate Turing kernelization for \textsc{Capacitated Dominating Set} with $\bigO(\frac{1+\eps}{\eps} \cdot \mathrm{tw} \cdot \Delta^2)$ vertices.
\end{theorem}
\begin{algorithm}[t]
	\SetAlgoLined
	\DontPrintSemicolon
	\SetKwInOut{Input}{input}\SetKwInOut{Output}{output}
	\Input{A graph $G=(V,E)$, nice tree decomposition $\calT$ of width $\mathrm{tw}$, $\eps>0$}
	$s \leftarrow 3\cdot\frac{1+\eps}{\eps} \cdot (\mathrm{tw}+1) \cdot (\Delta+1)^2$\;
	\eIf{$\abs{V} \leq s$}{
		Apply the $c$-approximate oracle to $G$ and output the result.\;\label{l:capds1}
	}{
		Use \cref{lemma:find-node} to find a node $t$ in $\calT$ such that $s \le \abs{V_t} \le 2s$.\;
		Apply the $c$-approximate oracle to $G[V_t]$ and let $(S_t,f_t)$ be the solution output by the oracle.\;
		Let $\calT'$ be the tree obtained by deleting the subtree rooted at $t$ except for the node $t$ from $\calT$.\;
		Apply this algorithm to $(G - (V_t\setminus X_t),\calT',\eps)$ and let $(S',f')$ be the returned solution.\;
		Apply \creflemmapart{lemma:capds}{lemma:capdsII} with $(X,f) = (S',f')$, $(Y,g) = (S_t,f_t)$, $A=(V\setminus V_t)\cup X_t$, $B=X_t$, and $C=V_t$. Let $(S,f)$ be the resulting solution for $G$.\;
		Return~$(S,f)$.\;\label{l:capds2}
	}
	
	\caption{A $(1+\eps)$-approximate polynomial Turing kernelization for \textsc{Capacitated Dominating Set} parameterized by $\mathrm{tw}+\Delta$}
	\label{alg:capds}
\end{algorithm}
\begin{proof}
	Consider \cref{alg:capds}.
	This algorithm always returns a capacitated dominating set of $G$.
	If the algorithm terminates in line~\ref{l:capds1}, then this is true because the oracle always outputs a capacitated dominating set.
	If it terminates in line~\ref{l:capds2}, then $(S_t,f_t)$ and $(S',f')$ are capacitated dominating sets for $G[V_t]$ and $G-(V_t\setminus X_t)$, respectively. It follows by \creflemmapart{lemma:capds}{lemma:capdsII}, that $(S,f)$ is a capacitated dominating set for $G$.
	
	The algorithm runs in polynomial time.
	
	Finally, we must show that the solution output by the algorithm contains at most $c\cdot (1+\eps)\cdot\OPTCAP(G)$ vertices.
	We prove the claim by induction on the number of recursive calls.
	If there is no recursive call, the algorithm terminates in line~\ref{l:capds1} and the solution contains at most $c\cdot\OPTCAP(G)$ vertices.
	Otherwise, by induction:
	\begin{align*}
	\abs{S} &\le  \abs{S'_t} + \abs{S'} + (\Delta +1)\cdot\abs{X_t} 
	\le c\cdot\OPTCAP(G[V_t]) + \abs{S'} + (\Delta +1)\cdot\abs{X_t}\\
	& =  c\cdot(1+\eps)\cdot\OPTCAP(G[V_t]) - c\cdot \eps \cdot \OPTCAP(G[V_t]) + \abs{S'} + (\Delta +1)\cdot\abs{X_t}\\
	& \stackrel{\dag}{\le} c\cdot(1+\eps)\cdot\OPTCAP(G[V_t]) - \frac{c\cdot \eps \cdot \abs{V_t}}{\Delta + 1} + \abs{S'} + (\Delta +1)\cdot\abs{X_t}\\ 
	& \le  c\cdot(1+\eps)\cdot\OPTCAP(G[V_t]) -  3\cdot c\cdot(1+\eps)\cdot(\mathrm{tw}+1)\cdot(\Delta+1) + \abs{S'} + (\Delta +1)\cdot\abs{X_t}\\
	& \stackrel{\ddag}{\le}  c\cdot(1+\eps)\cdot\OPTCAP(G[V_t]) - 3\cdot c\cdot(1+\eps)\cdot\abs{X_t}\cdot(\Delta+1) + \abs{S'} \\
	& \quad + c\cdot (1+\eps)\cdot (\Delta +1)\cdot\abs{X_t}\\
	& \le  c\cdot(1+\eps)\cdot\OPTCAP(G[V_t]) - 2\cdot c\cdot(1+\eps)\cdot\abs{X_t}\cdot(\Delta+1) + \abs{S'} \\
	& \le  c\cdot(1+\eps)\cdot\OPTCAP(G[V_t]) - 2\cdot c\cdot(1+\eps)\cdot\abs{X_t} + c\cdot(1+\eps)\cdot\OPTCAP(G-V_t)\\
	& =  c\cdot(1+\eps)\cdot(\OPTCAP(G[V_t]) - 2\abs{X_t} + \OPTCAP(G-V_t))\\
	& \stackrel{\mathparagraph}{\le}  c\cdot(1 + \eps) \cdot \OPTCAP(G)
	\end{align*}
	The inequality marked $\dag$ follows from \creflemmapart{lemma:ds}{lemma:dsI} and the fact that $\OPTCAP(G) \ge \OPTDS(G)$.
	$\ddag$ follows from the fact that $c\cdot (1+\eps) \ge 1$ and
	$\mathparagraph$ from \creflemmapart{lemma:capds}{lemma:capdsI}.
\end{proof}

\subsection{Independent Dominating Set}

The next problem we consider is \textsc{Independent Dominating Set}.

\begin{lemma}
	\label{lemma:indds}
	Let $G = (V,E)$ be a graph with maximum degree $\Delta$ and $A,B,C \subseteq V$ such that $A\cup C = V$, $A\cap C = B$, and there are no edges from $A\setminus B$ to $C \setminus B$.
	\begin{enumerate}[(i)]
		\item\label{lemma:inddsI} If $X$ is an independent set in $G$, then there is an independent dominating set $X'$ that contains $X$ and $\abs{X'\setminus X}$ is at most the number of vertices not dominated by $X$,  and such a set $X'$ can be computed in polynomial time.
		\item\label{lemma:inddsII} $\OPTIND(G) \ge \OPTIND(G[A]) + \OPTIND(G[C]) - 2\abs{B}$.
		\item\label{lemma:inddsIII} Given independent dominating sets $X$ and $Y$ in $G[A]$ and $G[C]$, respectively, one can in construct in polynomial time an independent dominating set for $G$ of size at most $\abs{X} + \abs{Y} + (\Delta+1) \cdot \abs{B}$.
	\end{enumerate}
\end{lemma}
\begin{proof}
	\begin{enumerate}[(i)]
		\item If $X$ is a dominating set, then $X'\coloneqq X$.
		Otherwise, there is a vertex $v \in V \setminus N[X]$.
		We add $v$ to $X$ and continue.
		Observe that when $v$ is added to $X$, the latter remains an independent set.
		\item Let $X$ be an independent dominating set of size $\OPTIND(G)$ in $G$.
		Let $Y \coloneqq X \cap A$.
		Since $Y \subseteq X$, it follows that $Y$ is an independent set.
		Moreover, $Y$ dominates all vertices in $(A\setminus B) \cup (X \cap B)$, leaving at most $B \setminus X$ vertices undominated.
		We apply (\ref{lemma:capdsI}) to $Y$ and obtain $Y'$, an independent dominating set in $G[A]$ of size at most $\abs{X\cap A} + \abs{B} - \abs{X \cap B}$.
		We apply the same argument to $G[C]$ to obtain an independent dominating set $Z$ of size at most $\abs{X\cap C} + \abs{B} - \abs{X \cap B}$.
		It follows that:
		\begin{align*}
		\OPTIND(G) &= \abs{X} = \abs{X\cap A} + \abs{X\cap C} - \abs{X\cap B}\\
		&= \abs{Y} - \abs{B\setminus X} + \abs{Z} - \abs{B} - \abs{X \cap B}\\
		&\ge \abs{Y} + \abs{Z} - 2\abs{B}\\
		&\ge \OPTIND(G[A]) + \OPTIND(G[C]) - 2\abs{B}.
		\end{align*}
		\item $Z\coloneqq (X \cup Y) \setminus B$ is an independent set in $G$.
		Since $X \cup Y$ is a dominating set and at most $(\Delta+1)\cdot\abs{B}$ vertices can be dominated by vertices in $B$, it follows that $Z$ leaves at most that many vertices in $G$ undominated.
		Applying (\ref{lemma:capdsI}) to $Z$ yields an independent dominating set of size at most $\abs{X} + \abs{Y} + (\Delta+1)\cdot\abs{B}$.
		\qedhere
	\end{enumerate}
\end{proof}

The kernelization algorithm for \textsc{Independent Dominating Set} uses \cref{lemma:indds}, but is otherwise similar to the one for \textsc{Capacitated Dominating Set} and is deferred to the appendix.

\begin{theorem}
	\label{thm:ub-ind}
	For every $\eps > 0$, there is a $(1+\eps)$-approximate Turing kernelization for \textsc{Independent Dominating Set} with $\bigO(\frac{1+\eps}{\eps} \cdot \mathrm{tw} \cdot \Delta^2)$ vertices.
\end{theorem}
\begin{algorithm}[t]
	\SetAlgoLined
	\DontPrintSemicolon
	\SetKwInOut{Input}{input}\SetKwInOut{Output}{output}
	\Input{A graph $G=(V,E)$, nice tree decomposition $\calT$ of width $\mathrm{tw}$, $\eps>0$}
	$s \leftarrow \abs{V} \leq 3\cdot\frac{1+\eps}{\eps} \cdot (\mathrm{tw}+1) \cdot (\Delta+1)^2$\;
	\eIf{$\abs{V} \leq s$}{
		Apply the $c$-approximate oracle to $G$ and output the result.\;\label{l:indds1}
	}{
		Use \cref{lemma:find-node} to find a node $t$ in $\calT$ such that $s \le \abs{V_t} \le 2s$.\;
		Apply the $c$-approximate oracle to $G[V_t]$ and let $S_t$ be the solution output by the oracle.\;
		Let $\calT'$ be the tree obtained by deleting the subtree rooted at $t$ except for the node $t$ from $\calT$.\;
		Apply this algorithm to $(G - (V_t\setminus X_t),\calT',\eps)$ and let $S'$ be the returned solution.\;
		Apply \creflemmapart{lemma:indds}{lemma:inddsIII} with $X = S'$, $Y = S_t$, $A=(V\setminus V_t)\cup X_t$, $B=X_t$, and $C=V_t$. Let $S$ be the resulting solution for $G$.\;
		Return~$S$.\;\label{l:indds2}
	}
	\caption{A $(1+\eps)$-approximate polynomial Turing kernelization for \textsc{Independent Dominating Set} parameterized by $\mathrm{tw}+\Delta$}
	\label{alg:indds}
\end{algorithm}

\begin{proof}
	Consider \cref{alg:indds}.
	This algorithm always returns an independent dominating set of $G$.
	If the algorithm terminates in line~\ref{l:indds1}, then this is true because the oracle always outputs an independent dominating set.
	If it terminates in line~\ref{l:indds2}, then $S_t$ and $S'$ are independent dominating sets for $G[V_t]$ and $G-(V_t\setminus X_t)$, respectively. It follows by \creflemmapart{lemma:indds}{lemma:inddsIII}, that $S$ is an independent dominating set for $G$.
	
	The algorithm runs in polynomial time.
	
	Finally, we must show that the solution output by the algorithm contains at most $c\cdot (1+\eps)\cdot\OPTIND(G)$ vertices.
	We prove the claim by induction on the number of recursive calls.
	If there is no recursive call, the algorithm terminates in line~\ref{l:indds1} and the solution contains at most $c\cdot\OPTIND(G)$ vertices.
	Otherwise, by induction:
	\begin{align*}
	\abs{S} &\le \abs{S'_t} + \abs{S'} + (\Delta +1)\cdot\abs{X_t} 
	\le c\cdot\OPTIND(G[V_t]) + \abs{S'} + (\Delta +1)\cdot\abs{X_t}\\
	&= c\cdot(1+\eps)\cdot\OPTIND(G[V_t]) - c\cdot \eps \cdot \OPTIND(G[V_t]) + \abs{S'} + (\Delta +1)\cdot\abs{X_t}\\
	&\stackrel{\dag}{\le} c\cdot(1+\eps)\cdot\OPTIND(G[V_t]) - \frac{c\cdot \eps \cdot \abs{V_t}}{\Delta + 1} + \abs{S'} + (\Delta +1)\cdot\abs{X_t}\\ 
	&\le c\cdot(1+\eps)\cdot\OPTIND(G[V_t]) -  3\cdot c\cdot(1+\eps)\cdot(\mathrm{tw}+1)\cdot(\Delta+1) + \abs{S'} + (\Delta +1)\cdot\abs{X_t}\\
	&\stackrel{\ddag}{\le} c\cdot(1+\eps)\cdot\OPTIND(G[V_t]) - 3\cdot c\cdot(1+\eps)\cdot\abs{X_t}\cdot(\Delta+1)\\
	& \quad + \abs{S'} + c\cdot (1+\eps)\cdot (\Delta +1)\cdot\abs{X_t}\\
	&\le c\cdot(1+\eps)\cdot\OPTIND(G[V_t]) - 2\cdot c\cdot(1+\eps)\cdot\abs{X_t}\cdot(\Delta+1) + \abs{S'} \\
	& \le c\cdot(1+\eps)\cdot\OPTIND(G[V_t]) - 2\cdot c\cdot(1+\eps)\cdot\abs{X_t} + c\cdot(1+\eps)\cdot\OPTIND(G-V_t)\\
	& = c\cdot(1+\eps)\cdot(\OPTIND(G[V_t]) - 2\abs{X_t} + \OPTIND(G-V_t))\\
	&\stackrel{\mathparagraph}{\le} c\cdot(1 + \eps) \cdot \OPTIND(G)
	\end{align*}
	The inequality marked with $\dag$ follows from \creflemmapart{lemma:ds}{lemma:dsI} and the fact that $\OPTIND(G) \ge \OPTDS(G)$.
	$\ddag$ follows from the fact that $c\cdot (1+\eps) \ge 1$ and 	$\mathparagraph$ from \creflemmapart{lemma:indds}{lemma:inddsII}.
\end{proof}

\subsection{Connected Dominating Set}
\label{sec:cds}
Finally, we consider the problem \textsc{Connected Dominating Set}.
If $S\subseteq V$ is a vertex set in a graph $G=(V,E)$, then let $R(G,S)$ denote the graph obtained by deleting $S$, introducing a new vertex $z$, and connecting $z$ to any vertex in $V\setminus S$ that has a neighbor in $S$.

\begin{lemma}
	\label{lemma:conds}
	Let $G = (V,E)$ be a connected graph and $A,B,C \subseteq V$ such that $A\cup C = V$, $A\cap C = B$, there are no edges from $A\setminus B$ to $C \setminus B$, and $A\setminus B$ and $C\setminus B$ are both non-empty.
	\begin{enumerate}[(i)]
		\item\label{lemma:condsI} $\OPTCON(G) \ge \OPTCON(R(G[A],B)) + \OPTCON(R(G[C],B)) - 2$.
		\item\label{lemma:condsII} Given connected dominating sets $X$ and $Y$ in $R(G[A],B)$ and $R(G[C],B)$, respectively, one can in construct in polynomial time a connected dominating set for $G$ of size at most $\abs{X} + \abs{Y} + 3\abs{B}$.
	\end{enumerate}
\end{lemma}
\begin{proof}
	\begin{enumerate}[(i)]
		\item Let $X$ be a connected dominating set in $G$ of size $\OPTCON(G)$.
		We claim that $Y \coloneqq (X\cap (A \setminus B)) \cup \{z\}$ and $Z \coloneqq (X\cap (C\setminus B)) \cup \{z\}$ are connected dominating sets in $R(G[A],B)$ and $R(G[C],B)$, respectively.
		We only prove this for $Y$ and $R(G[A],B))$, as the case of $Z$ and $R(G[C],B)$ is analogous.
		
		First, we show that $Y$ is a dominating set.
		Let $v$ be a vertex in $R(G[A],B))$.
		If $v \in \{z,z'\}$, then $v$ is dominated by $z$ in $Y$.
		Otherwise, $v \in A \setminus B$ and there is a vertex $w \in X$ that dominates $v$ in $G$.
		If $w \in B$, then $z$ is adjacent to $v$ in $R(G[A],B))$ and $v$ is dominated by $z$ in that graph.
		If $w \in A \setminus B$, then $w \in Y$ and $v$ is dominated by $w$ in $R(G[A],B))$.
		
		We must also show that the subgraph of $R(G[A],B)$ induced by $Y$ is connected.
		Let $v,v' \in Y$.
		We must show that the subgraph of $R(G[A],B)$ induced by $Y$ contains a path from $v$ to $v'$.
		First we assume that $v,v' \ne z,z'$, implying that $v,v' \in X$.
		Hence, there is a path $P$ from $v$ to $v'$ in $G[X]$.
		If $P \subseteq A \setminus B$, then $P \subseteq Y$ and we are done.
		Otherwise, $P$ must pass through $B$.
		Let $w$ be the first vertex in $B$ on $P$ and $w'$ the final one.
		Obtain $P'$ by replacing the subpath of $P$ between and including $w$ and $w'$ with $z$.
		Then, $P'$ is a path from $v$ to $v'$ in the subgraph of $R(G[A],B)$ induced by $Y$.
		Finally, suppose that $v' = z$.
		If $v= z$, there is nothing to show, so we assume that $v\ne z$ and, therefore, $v\in X\cap (A\setminus B)$.
		Because $C\setminus B$ is non-empty, $X$ must contain a vertex $w\in B$.
		Because $G[X]$ is connected, $G[X]$ must also contain a path $P$ from $v$ to $w$.
		Let $w'$ be the first vertex in $B$ on the path $P$ (possibly, $w'=w$).
		We obtain $P'$, a path from $v$ to $z$ in the subgraph of $R(G[A],B)$ induced by $Y$, by taking the subpath of $P$ from $v$ to $w'$ and replacing $w'$ with $z$.
		This proves that $Y$ is a connected dominating set in $R(G[A],B)$.
		Then,
		\begin{align*}
		\OPTCON(G) &= \abs{X} \ge \abs{X\cap (A\setminus B)} + \abs{X\cap (C\setminus B)} \\
		&\ge \abs{Y} - 1 + \abs{Z} - 1\\
		&\ge \OPTCON(R(G[A],B)) + \OPTCON(R(G[C],B)) - 2.
		\end{align*}
		\item 
		Let $Z' \coloneqq X \cup Y \cup B$.
		Every connected component of $G[Z']$ contains a vertex in $B$, so this graph as at most $\abs{B}$ connected components.
		We obtain a connected dominating set $Z$ in $G$ as follows.
		We start with $Z \coloneqq Z'$.
		Choose two connected components $C_1,C_2$ in $G[Z]$.
		Because $G$ is connected, it contains a path $P$ starting in $v_1 \in C_1$ and ending in $v_2 \in C_2$.
		This path must contain a vertex that is not adjacent to any vertex in $C_1$, because if every vertex in $P\setminus C_1$ were adjacent to a vertex in $C_1$, then $v_2$ is adjacent to a vertex in $C_1$, implying that $C_1$ and $C_2$ are not distinct connected components in $G[Z]$
		Let $w$ be the first vertex on $P$ that is not adjacent to a vertex in $C_1$.
		Because $Z$ is a dominating set in $G$, there must be a vertex $x \in Z \setminus C_1$ such that $w \in N[x]$ (note that, possible $w=x$).
		Adding $w$ and $x$ merges $C_1$ with the connected component of $G[Z]$ containing $x$.
		This process must be repeated at most $\abs{B}$ times to obtain a connected dominating set.
		In each iteration at most two vertices are added to $Z$.
		Since $Z$ initially contains $\abs{X} + \abs{Z} + \abs{B}$ vertices, we obtain a connected dominating set containing at most $\abs{X} + \abs{Z} + 3\abs{B}$ vertices.\qedhere
	\end{enumerate}
\end{proof}

The kernelization algorithm for \textsc{Connected Dominating Set} is similar to the one for \textsc{Capacitated Dominating Set} and is deferred to the appendix.

\begin{theorem}
	\label{thm:ub-con}
	For every $\eps > 0$, there is a $(1+\eps)$-approximate Turing kernelization for \textsc{Connected Dominating Set} with $\bigO(\frac{1+\eps}{\eps} \cdot \mathrm{tw} \cdot \Delta)$ vertices.
\end{theorem}

\begin{algorithm}[t]
	\SetAlgoLined
	\DontPrintSemicolon
	\SetKwInOut{Input}{input}\SetKwInOut{Output}{output}
	\Input{A graph $G=(V,E)$, nice tree decomposition $\calT$ of width $\mathrm{tw}$, $\eps>0$}
	$s \leftarrow  4\cdot \frac{1+\eps}{\eps}(\Delta+1)(\mathrm{tw}+1) + \frac{(2 \Delta+2)(1+\eps)}{\eps}$\;
	\eIf{$\abs{V} \leq s$}{
		Apply the $c$-approximate oracle to $G$ and output the result.\;
		\label{l:conds1}
	}{
		Use \cref{lemma:find-node} to find a node $t$ in $\calT$ such that $s \le \abs{V_t} \le 2s$.\;
		Apply the $c$-approximate oracle to $R(G[V_t],X_t)$ and let $S_t$ be the solution output by the oracle..\;
		Let $\calT'$ be the tree obtained by deleting the subtree rooted at $t$ except for the node $t$ from $\calT$.\;
		Apply this algorithm to $(R(G - (V_t\setminus X_t),X_t),\calT',\eps)$ and let $S'$ be the returned solution.\;
		Apply \creflemmapart{lemma:conds}{lemma:condsII} with $X = S'$, $Y = S_t$, $A=(V\setminus V_t)\cup X_t$, $B=X_t$, and $C=V_t$. Let $S$ be the resulting solution for $G$.\;
		Return~$S$.\;\label{l:conds2}	
	}
	
	\caption{A $(1+\eps)$-approximate polynomial Turing kernelization for \textsc{Connected Dominating Set} parameterized by $\mathrm{tw}+\Delta$}
	\label{alg:conds}
\end{algorithm}

\begin{proof}
	Consider \cref{alg:conds}. This algorithm always returns a connected dominating set of $G$.
	If the algorithm terminates in line~\ref{l:conds1}, then this is true because the oracle always outputs a connected dominating set.
	If it terminates in line~\ref{l:conds2}, then $S_t$ and $S'$ are connected dominating sets for $R(G[V_t],X_t)$ and $R(G - (V_t\setminus X_t))$, respectively. It follows by \creflemmapart{lemma:conds}{lemma:condsII}, that $S$ is a connected dominating set for $G$.
	
	The algorithm runs in polynomial time.
	
	Finally, we must show that the solution output by the algorithm contains at most $c\cdot (1+\eps)\cdot\OPTCON(G)$ vertices.
	We prove the claim by induction on the number of recursive calls.
	If there is no recursive call, the algorithm terminates in line~\ref{l:conds1} and the solution contains at most $c\cdot\OPTCAP(G)$ vertices.
	Otherwise, by induction:
	\begin{align*}
	\abs{S} &\le \abs{S'_t} + \abs{S'} + 3\abs{X_t} 
	\le c\cdot\OPTCON(R(G[V_t],X_t)) + \abs{S'} + 3\abs{X_t}\\
	&= c\cdot(1+\eps)\cdot\OPTCON(R(G[V_t],X_t)) - c\cdot \eps \cdot \OPTCON(R(G[V_t],X_t)) + \abs{S'} + 3\abs{X_t}\\
	&\stackrel{\dag}{\le} c\cdot(1+\eps)\cdot\OPTCON(R(G[V_t],X_t)) - \frac{c\cdot \eps \cdot (\abs{V_t}-X_t+1)}{\Delta + 1} + \abs{S'} + 3\abs{X_t}\\
	&= c\cdot(1+\eps)\cdot\OPTCON(R(G[V_t],X_t)) - \frac{c\cdot \eps \cdot \abs{V_t}}{\Delta + 1}  + \abs{S'}  + \frac{c\cdot\eps\cdot(\abs{X_t}+1)}{\Delta+1} + 3\abs{X_t}\\
	&\le c\cdot(1+\eps)\cdot\OPTCON(G[V_t]) -  4\cdot c\cdot(1+\eps)\cdot(\mathrm{tw}+2) - 2c(1+\eps) + \abs{S'}\\
	& \quad + \frac{c\cdot\eps\cdot(\abs{X_t}+1)}{\Delta+1} + 3\abs{X_t}\\
	&\stackrel{\ddag}{\le} c\cdot(1+\eps)\cdot\OPTCON(G[V_t]) - 4\cdot c\cdot(1+\eps)\cdot(\abs{X_t}+1) - 2c(1+\eps) + \abs{S'} \\
	& \quad  + c\cdot (1+\eps)\cdot(4\abs{X_t} + 1)\\
	&\le c\cdot(1+\eps)\cdot\OPTCON(G[V_t]) - 2\cdot c(1+\eps) + \abs{S'} \\
	& \le c\cdot(1+\eps)\cdot\OPTCON(G[V_t]) - 2c(1+\eps) + c\cdot(1+\eps)\cdot\OPTCON(G-V_t)\\
	& = c\cdot(1+\eps)\cdot(\OPTCON(G[V_t]) - 2 + \OPTCON(G-V_t))\\
	&\stackrel{\mathparagraph}{\le} c\cdot(1 + \eps) \cdot \OPTCON(G)
	\end{align*}
	The inequality marked with $\dag$ follows from \creflemmapart{lemma:ds}{lemma:dsII} and the fact that $\OPTCON(G) \ge \OPTDS(G)$.
	$\ddag$ follows from the fact that $c\cdot (1+\eps) \ge 1$ and 
	$\mathparagraph$ from \creflemmapart{lemma:conds}{lemma:condsI}.
\end{proof}

\section{Conclusion}
We conclude by pointing out two open problems concerning approximate Turing kernelization:
\begin{itemize}
	\item Does \textsc{Connected Feedback Vertex Set} parameterized by treewidth admit an approximate polynomial Turing kernelization?
	The approach employed by Hols et al.~\cite{Hols2020} for \textsc{Connected Vertex Cover} and here for \textsc{Connected Dominating Set} cannot be used for \textsc{Connected Feedback Vertex Set}, because the ratio between the size of a minimum connected feedback vertex and the size of a minimum feedback vertex set is unbounded.
	\item The biggest open question in Turing kernelization is whether or not there are polynomial Turing kernelizations for the problems \textsc{Longest Path} and \textsc{Longest Cycle} parameterized by the solution size~\cite{Hermelin2013}.
	There has been some progress on this problem by considering the restriction to certain graph classes~\cite{Jansen2017,Jansen2018}.
	Developing an \emph{approximate} Turing kernelization may be another way of achieving progress in this regard.
\end{itemize}

\bibliography{strings-long,tk}

\begin{thebibliography}{10}

\bibitem{Binkele2012}
Daniel Binkele-Raible, Henning Fernau, Fedor~V. Fomin, Daniel Lokshtanov, Saket
  Saurabh, and Yngve Villanger.
\newblock Kernel(s) for problems with no kernel: On out-trees with many leaves.
\newblock {\em ACM Transactions on Algorithms}, 8(4), 2012.
\newblock \href {https://doi.org/10.1145/2344422.2344428}
  {\path{doi:10.1145/2344422.2344428}}.

\bibitem{Chlebik2008}
M.~Chlebík and J.~Chlebíková.
\newblock Approximation hardness of dominating set problems in bounded degree
  graphs.
\newblock {\em Information and Computation}, 206(11):1264--1275, 2008.
\newblock \href {https://doi.org/10.1016/j.ic.2008.07.003}
  {\path{doi:10.1016/j.ic.2008.07.003}}.

\bibitem{Chvatal1979}
Vasek Chv{\'{a}}tal.
\newblock A greedy heuristic for the set-covering problem.
\newblock {\em Mathematics of Operations Research}, 4(3):233--235, 1979.
\newblock \href {https://doi.org/10.1287/moor.4.3.233}
  {\path{doi:10.1287/moor.4.3.233}}.

\bibitem{Dom2014}
Michael Dom, Daniel Lokshtanov, and Saket Saurabh.
\newblock Kernelization lower bounds through colors and {IDs}.
\newblock {\em ACM Transactions on Algorithms}, 11(2):1--20, 2014.
\newblock \href {https://doi.org/10.1145/2650261} {\path{doi:10.1145/2650261}}.

\bibitem{Feige2008}
Uriel Feige, MohammadTaghi Hajiaghayi, and James~R. Lee.
\newblock Improved approximation algorithms for minimum weight vertex
  separators.
\newblock {\em SIAM Journal on Computing}, 38(2):629--657, 2008.
\newblock \href {https://doi.org/10.1137/05064299X}
  {\path{doi:10.1137/05064299X}}.

\bibitem{Fellows2018}
Michael~R. Fellows, Ariel Kulik, Frances Rosamond, and Hadas Shachnai.
\newblock Parameterized approximation via fidelity preserving transformations.
\newblock {\em Journal of Computer and System Sciences}, 93:30--40, 2018.
\newblock \href {https://doi.org/10.1016/j.jcss.2017.11.001}
  {\path{doi:10.1016/j.jcss.2017.11.001}}.

\bibitem{Fomin2019}
Fedor~V Fomin, Daniel Lokshtanov, Saket Saurabh, and Meirav Zehavi.
\newblock {\em Kernelization: Theory of Parameterized Preprocessing}.
\newblock Cambridge University Press, 2019.
\newblock \href {https://doi.org/10.1017/9781107415157}
  {\path{doi:10.1017/9781107415157}}.

\bibitem{Guha1998}
Sudipto Guha and Samir Khuller.
\newblock Approximation algorithms for connected dominating sets.
\newblock {\em Algorithmica}, 20:374--387, 1998.
\newblock \href {https://doi.org/10.1007/PL00009201}
  {\path{doi:10.1007/PL00009201}}.

\bibitem{Halldorsson1993}
Magnús~M. Halldórsson.
\newblock Approximating the minimum maximal independence number.
\newblock {\em Information Processing Letters}, 46(4):169--172, 1993.
\newblock \href {https://doi.org/10.1016/0020-0190(93)90022-2}
  {\path{doi:10.1016/0020-0190(93)90022-2}}.

\bibitem{Heggernes2013}
Pinar Heggernes, Pim {van ’t Hof}, Bart~M.P. Jansen, Stefan Kratsch, and
  Yngve Villanger.
\newblock Parameterized complexity of vertex deletion into perfect graph
  classes.
\newblock {\em Theoretical Computer Science}, 511:172--180, 2013.
\newblock \href {https://doi.org/10.1016/j.tcs.2012.03.013}
  {\path{doi:10.1016/j.tcs.2012.03.013}}.

\bibitem{Hermelin2013}
Danny Hermelin, Stefan Kratsch, Karolina So{\l}tys, Magnus Wahlstr{\"o}m, and
  Xi~Wu.
\newblock A completeness theory for polynomial ({Turing}) kernelization.
\newblock In {\em Proceedings of the 8th International Symposium on
  Parameterized and Exact Computation (IPEC)}, pages 202--215, 2013.
\newblock \href {https://doi.org/10.1007/978-3-319-03898-8_18}
  {\path{doi:10.1007/978-3-319-03898-8_18}}.

\bibitem{Hols2020}
Eva-Maria~C. Hols, Stefan Kratsch, and Astrid Pieterse.
\newblock Approximate {Turing} kernelization for problems parameterized by
  treewidth.
\newblock In {\em Proceedings of the 28th Annual European Symposium on
  Algorithms (ESA)}, pages 60:1--60:23, 2020.
\newblock \href {https://doi.org/10.4230/LIPIcs.ESA.2020.60}
  {\path{doi:10.4230/LIPIcs.ESA.2020.60}}.

\bibitem{Impagliazzo2001}
Russell Impagliazzo and Ramamohan Paturi.
\newblock On the complexity of {$k$-SAT}.
\newblock {\em Journal of Computer and System Sciences}, 62(2):367--375, 2001.
\newblock \href {https://doi.org/10.1006/jcss.2000.1727}
  {\path{doi:10.1006/jcss.2000.1727}}.

\bibitem{Irving1991}
Robert~W. Irving.
\newblock On approximating the minimum independent dominating set.
\newblock {\em Information Processing Letters}, 37(4):197--200, 1991.
\newblock \href {https://doi.org/10.1016/0020-0190(91)90188-N}
  {\path{doi:10.1016/0020-0190(91)90188-N}}.

\bibitem{Jansen2017}
Bart M.~P. Jansen.
\newblock Turing kernelization for finding long paths and cycles in restricted
  graph classes.
\newblock {\em Journal of Computer and System Sciences}, 85:18--37, 2017.
\newblock \href {https://doi.org/10.1016/j.jcss.2016.10.008}
  {\path{doi:10.1016/j.jcss.2016.10.008}}.

\bibitem{Jansen2018}
Bart M.~P. Jansen, Marcin Pilipczuk, and Marcin Wrochna.
\newblock Turing kernelization for finding long paths in graphs excluding a
  topological minor.
\newblock In {\em Proceedings of the 12th International Symposium on
  Parameterized and Exact Computation (IPEC)}, pages 23:1--23:13, 2018.
\newblock \href {https://doi.org/10.4230/LIPIcs.IPEC.2017.23}
  {\path{doi:10.4230/LIPIcs.IPEC.2017.23}}.

\bibitem{Klein1995}
Philip Klein and R.~Ravi.
\newblock A nearly best-possible approximation algorithm for node-weighted
  {Steiner} trees.
\newblock {\em Journal of Algorithms}, 19(1):104--115, 1995.
\newblock \href {https://doi.org/10.1006/jagm.1995.1029}
  {\path{doi:10.1006/jagm.1995.1029}}.

\bibitem{Kloks1994}
Ton Kloks.
\newblock {\em Treewidth: Computations and Approximations}.
\newblock Springer, 1994.
\newblock \href {https://doi.org/10.1007/BFb0045375}
  {\path{doi:10.1007/BFb0045375}}.

\bibitem{Lokshtanov2008}
Daniel Lokshtanov.
\newblock Wheel-free deletion is {W[2]}-hard.
\newblock In {\em Proceedings of the 3rd International Symposium on
  Parameterized and Exact Computation (IPEC)}, pages 141--147, 2008.
\newblock \href {https://doi.org/10.1007/978-3-540-79723-4_14}
  {\path{doi:10.1007/978-3-540-79723-4_14}}.

\bibitem{Lokshtanov2016}
Daniel Lokshtanov, Fahad Panolan, M.~S. Ramanujan, and Saket Saurabh.
\newblock Lossy kernelization.
\newblock {\em CoRR}, abs/1604.04111, 2016.
\newblock Full version of \cite{Lokshtanov2017}.
\newblock \href {http://arxiv.org/abs/1604.04111} {\path{arXiv:1604.04111}},
  \href {https://doi.org/10.48550/arXiv.1604.04111}
  {\path{doi:10.48550/arXiv.1604.04111}}.

\bibitem{Lokshtanov2017}
Daniel Lokshtanov, Fahad Panolan, M.~S. Ramanujan, and Saket Saurabh.
\newblock Lossy kernelization.
\newblock In {\em Proceedings of the 49th Annual ACM Symposium on Theory of
  Computing (STOC 2017)}, pages 224--237, 2017.
\newblock \href {https://doi.org/10.1145/3055399.3055456}
  {\path{doi:10.1145/3055399.3055456}}.

\bibitem{Nelson2007}
Jelani Nelson.
\newblock A note on set cover inapproximability independent of universe size.
\newblock {\em Electronic Colloquium on Computational Complexity}, {TR07-105},
  2007.
\newblock URL:
  \url{https://eccc.weizmann.ac.il/eccc-reports/2007/TR07-105/index.html}.

\bibitem{Wolsey1982}
Laurence~A. Wolsey.
\newblock An analysis of the greedy algorithm for the submodular set covering
  problem.
\newblock {\em Combinatorica}, 2(4):385--393, 1982.
\newblock \href {https://doi.org/10.1007/BF02579435}
  {\path{doi:10.1007/BF02579435}}.

\end{thebibliography}

\clearpage

\end{document}